\documentclass[12pt]{article}

\usepackage{amsmath}
\usepackage{amsfonts}
\usepackage{graphicx}

\usepackage[title]{appendix}

\usepackage{amssymb}
\usepackage{amscd}
\usepackage{amsxtra}
\usepackage{amsfonts}
\usepackage{amsthm}
\usepackage{framed}

\usepackage{ulem}
\newcommand{\playerSet}{\mathcal{V}}

\newtheorem{theorem}{Theorem}
\title{Self-regulation promotes cooperation in social networks}

\author
{Dario Madeo,$^{1}$ Chiara Mocenni$^{1\ast}$\\%$^{1\ast}$
\\
\normalsize{$^{1}$Department of Information Engineering and Mathematics, University of Siena,}\\
\normalsize{Via Roma 56, 53100, Siena, Italy.}\\
\\
\normalsize{$^\ast$To whom correspondence should be addressed; E-mail:  mocenni@dii.unisi.it.}
}

\date{}

\begin{document} 


%\baselineskip24pt


\maketitle


\begin{abstract}
Cooperative behavior in real social dilemmas is often perceived 
as a phenomenon emerging from norms and punishment.
To overcome this paradigm,
we highlight the interplay between the influence of social networks on individuals,
and the activation of spontaneous self-regulating mechanisms,
which may lead them to behave cooperatively,
while interacting with others and taking conflicting decisions over time. % without central control.
By extending Evolutionary game theory over networks, we prove that cooperation partially or fully emerges whether self-regulating mechanisms
are sufficiently stronger than social pressure.
Interestingly, even few cooperative individuals
act as catalyzing agents for the cooperation of others,
thus activating a recruiting mechanism, eventually driving the whole population 
to cooperate.\\
\end{abstract}

%\noindent\textbf{Sentence Summary:} 
%Cooperation in social groups results from the interplay between conflicts among individuals
%and self-regulating mechanisms on networks.\\

% In setting up this template for *Science* papers, we've used both
% the \section* command and the \paragraph* command for topical
% divisions.  Which you use will of course depend on the type of paper
% you're writing.  Review Articles tend to have displayed headings, for
% which \section* is more appropriate; Research Articles, when they have
% formal topical divisions at all, tend to signal them with bold text
% that runs into the paragraph, for which \paragraph* is the right
% choice.  Either way, use the asterisk (*) modifier, as shown, to

Cooperation in human populations is a fundamental phenomenon,
which has fascinated many scientists working in different fields, such as biology, sociology, economics, 
\cite{Hammerstein2003, Nowak2004SCIENCE,DOEBELI2004,Gintis2005, Pennisi2009,RandNowak2013},
and, more recently, engineering \cite{CAO2013,TOUPO2015,LEONARD2018}.
Actually, in the analysis of social dilemmas, such as the prisoner's dilemma game, cooperation is often assumed to be 
the result of suitable norms \cite{FEHR2004,HAUERT2007} and punishment mechanisms 
\cite{FEHR2002, SIGMUND2003, BOYD2010, HELBING2010, Nowak2011, STANLEY2017}.
Further approaches claim that cooperation emerges from other factors, such as 
reputation \cite{Nowak2006_2, WANG2017},
synergy and discounting \cite{CH2006},
social diversity \cite{SANTOS2008} and positive interactions \cite{Nowak2009PublicCooperation}.\\
Mathematical models tackling the problem of cooperation are typically grounded on evolutionary game theory, 
where a population of individuals taking part in a game repeated over time is considered.
Each individual decides to be Cooperative ($C$) or Defective ($D$)
when playing one or more two-players games with other individuals.
Eventually he obtains a given payoff,
and may decide to opt for changing his strategy.
In a generic two-players game, $R$ is the reward when both cooperate, $T$ is the temptation
to defect when the opponent cooperates, $S$ is the sucker's payoff earned by
a cooperative player when the opponent is a free rider, and $P < R$ is the punishment
for mutual defection.
The social dilemma arises when the temptation to defect is stronger than the reward for cooperation ($T > R$),
and the punishment for defection is preferred to the sucker's payoff ($P > S$).
This scheme is known as prisoner's dilemma game, and it well describes
many real world situations, occuring in fields like 
environmental protection, biology, psychology, international politics, economics and so on.
In all these cases,
defectors earn a higher payoff than unconditional
cooperators. As a consequence, cooperators
are doomed to extinction. This fact has been proven to be true 
for both the deterministic setting of the replicator equation and for stochastic
game dynamics of finite populations, assuming that
players are equally likely to interact with each other \cite{HOF1998, Nowak2004, Nowak2006}.
Interestingly, a consistent part of the recent literature
\cite{SANTOS2005, Nowak2005, SANTOS2006, NOWAK2006_REGRAPH, PIN2015, Adami2016, Nowak2017, GG2012} 
showed that the networked structure of the population may favor the emergence of cooperation.\\

In this work, we consider members of a social network.
To be realistic, 
these individuals may choose to be $100\%$ cooperative, $100\%$ defective, as in the standard case,
or also fuzzy strategies, such as $50\%-50\%$ cooperative/defective.
More interestingly, in this study cooperation is shown to be promoted by spontaneous self-regulating mechanisms,
according to the idea that humans seem to have an innate tendency to
cooperate with one another even when it
goes against their rational self-interest, as pointed out by 
Vogel in \cite{VOGEL2004}.
Specifically, self-regulation is modeled as an inertial term and the resulting effects on the population dynamics
are extensively analyzed both theoretically and by means of numerical simulations.\\

%In this work, we consider members of a networked population.
%To be more realistic, we also allow them to choose 
%more fuzzy strategies, such as $50\%-50\%$ cooperative/defective, 
%instead of $100\%$ cooperative ($C$) or $100\%$ defective ($D$) strategies, as in the standard case.
%%
%Moreover, an inertial term representing self-regulation and awareness of agents
%is introduced, and the resulting effects on the population dynamics
%are extensively analyzed from a theoretical and numerical points of view.
%\\
 
Consider a social network defined as a finite population of players $v = \{1, \ldots, N\}$, 
arranged on a directed graph described by adjacency matrix ${\bf A} = \{a_{v,w}\}$,
where $a_{v,w} = 1$ if player $v$ is influenced by player $w$, and $0$ otherwise.
The in-degree $k_{v} = \sum_{w=1}^{N}a_{v,w}$ of a player represents the number 
of neighbors of the generic player $v$. 
At each time instant a player can 
choose his own level of cooperation, indicated by $x_{v} \in [0,1]$,
to play $k_{v}$ two-persons games
with all his neighbors.
Notice that $x_{v} = 1$ represents 
full cooperation ($C$), 
$x_{v} = 0$ means full defection ($D$),
and $x_{v} \in (0, 1)$ stands for a partial level of cooperation.
Accordingly, when any two connected players $v$ and $w$ take part in a game,
the payoff for $v$ is defined by the continuous function $\phi : [0, 1] \times [0, 1] \to \mathbb{R}$:
\begin{equation}\label{eqn:phiTwoPlayers}	
\phi(x_{v}, x_{w}) = [x_{v}, 1-x_{v}]\begin{bmatrix}
R & S\\
T & P
\end{bmatrix} \begin{bmatrix}
x_{w} \\
1-x_{w}
\end{bmatrix}.
\end{equation}
Notice that for $x_v, x_w \in \{0,1\}$, we recover the payoffs $R$, $T$, $S$ and $P$ previously introduced.
Player's $v$ payoff $\phi_{v}$ collected over the network is the sum of all outcomes of two-player games with neighbors.
Formally, the payoff function $\phi_{v} : [0, 1]^{N} \to \mathbb{R}$ is defined as follows:
\begin{equation}\label{eqn:phiNetwork}
	%\phi_{v}(x_{1}, \ldots, x_{N}) = \sum_{w = 1}^{N} a_{v,w}\phi(x_{v}, x_{w}),
	\phi_{v}({\bf x}) = \sum_{w = 1}^{N} a_{v,w}\phi(x_{v}, x_{w}),
\end{equation}
where ${\bf x} = [x_{1}, x_{2}, \ldots, x_{N}]^{\top}$.
Player $v$ is able to appraise whether a change of his strategy $x_{v}$
produces an improvement of his payoff $\phi_{v}$. 
Indeed, if the derivative
of $\phi_{v}$ with respect to $x_{v}$ is positive (negative), 
the player would like to increase (decrease) his level of cooperation.
Of course, when the derivative is null, then the player has no incentive
to change his mind.
This mechanism is modeled by the EGN equation \cite{MM2015,IMM2016}, which reads as follows:
\begin{equation}\label{eqn:EGN-base}
\dot{x}_{v} = x_{v}(1-x_{v})\displaystyle \frac{\partial \phi_{v}}{\partial x_{v}},
\end{equation}
where the sign of $\dot{x}_{v}$ depends only on the term $\partial \phi_{v}/\partial x_{v}$,
since $x_{v}(1-x_{v}) \geq 0$.
Unlike the standard replicator equation, which
deals with the distribution of strategies over a well mixed population,
where players are indistinguishable except for their strategies, 
the EGN is a system of ODEs, each one describing the
strategy evolution of the specific player $v$
(when dealing with more than two strategies, the dimension of the ODE system increases accordingly).
%This fact is fundamental for our findings, since it allows us to provide
%precise conditions for each player to be cooperative. 
%
Following equation \eqref{eqn:phiNetwork}, the derivative of the payoff $\phi_{v}$ is:
$$\frac{\partial \phi_{v}}{\partial x_{v}} = \sum_{w=1}^{N}a_{v,w} \frac{\partial \phi(x_{v},x_{w})}{\partial x_{v}}.$$
Moreover, from equation \eqref{eqn:phiTwoPlayers} we get that:
\begin{equation}\label{eqn:phiDerivative}
	\frac{\partial \phi(x_{v},x_{w})}{\partial x_{v}} = (R-T+P-S)x_{w} - (P-S) = (\sigma_{C}+\sigma_{D})x_{w} - \sigma_{D},
\end{equation}
where $\sigma_{C} = R-T$ and $\sigma_{D} = P-S$ \cite{WEIBULL1995}.
For the specific case of a strict prisoner's dilemma game, $\sigma_{C} < 0$ and $\sigma_{D} > 0$. 
According to \cite{SANTOS2006} and \cite{Wang2015},  
unilateral defection is preferred to mutual cooperation when $\sigma_{C}<0$,
while $\sigma_{D} > 0$ indicates the preference for mutual defection
over unilateral cooperation.
When the effect of $\sigma_{C}$ is stronger than $\sigma_{D}$ ($|\sigma_{C}| > \sigma_{D}$),
the game is more influenced by temptation (hereafter called {\it T-driven} game),
while in the other case ($|\sigma_{C}| < \sigma_{D}$)
punishment is more effective, and hence the game will be called {\it P-driven}.
Since $\sigma_{C} < 0$ and $\sigma_{D} > 0$, we have that $\partial \phi(x_{v},x_{w})/\partial x_{v} \leq 0$
for all $x_{w} \in [0, 1]$. Therefore, also $\partial \phi_{v}/\partial x_{v} \leq 0$
and then $\dot{x}_v \leq 0,$
showing that the level of cooperation decreases over time towards full defection.\\

As reported in equation \eqref{eqn:EGN-base}, $\partial \phi_{v}/\partial x_{v}$ depends only on the state of neighboring players,
not on the current state $x_v$ of player $v$ himself.
On the contrary, the willingness to pursue cooperation as a greater good follows from  
internal mechanisms correlated to personal awareness and culture. 
These mechanisms should act as inertial terms able to reduce the rational temptation to defect,
depending on the current strategy $x_{v}$ of player himself. Generally, this aspect
is not taken into account in the mathematical modeling, even though it normally characterizes complex 
individuals like humans \cite{VOGEL2004}.\\

\begin{figure}[!ht]
\begin{center}
\includegraphics[width=0.8\columnwidth]{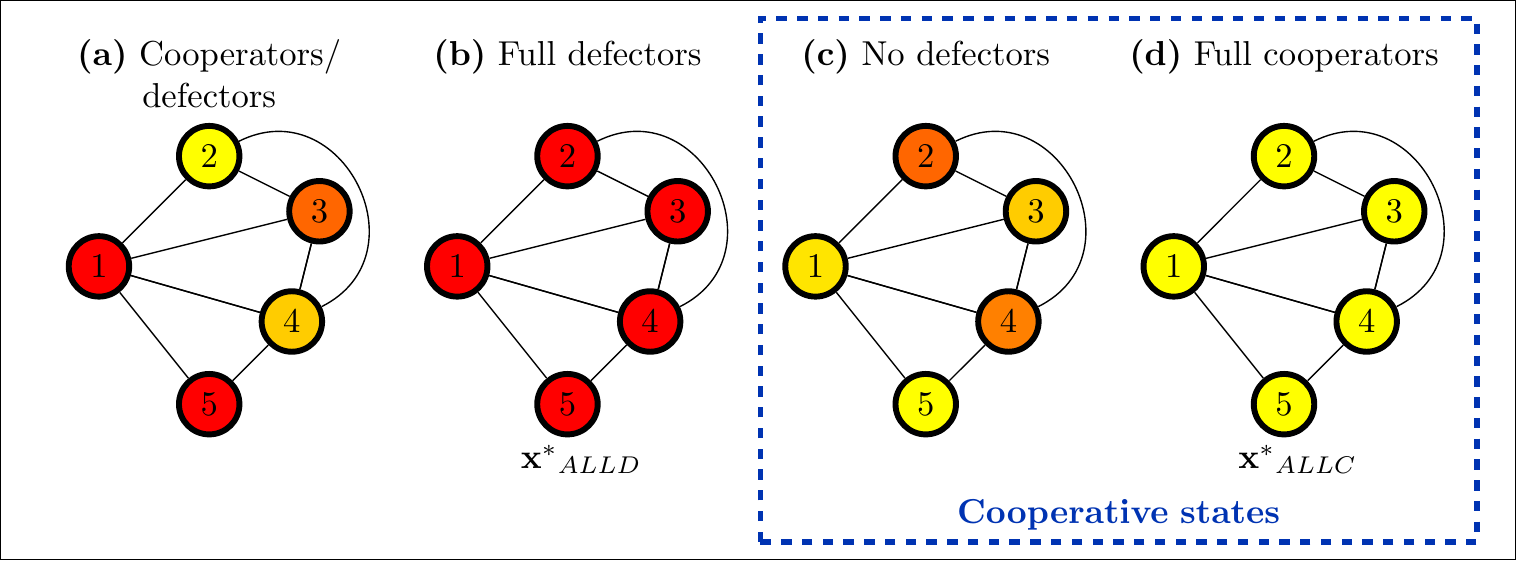}
\caption{\footnotesize{\textbf{Cooperation distribution over a social network.} 
\textbf{(a)} Coexistence of full cooperators (yellow circles), partial cooperators (orange circles) and 
full defectors (red circles).
\textbf{(b)} Population composed only by full defectors. This state is
the $N$-dimensional vector ${\bf x}^{*}_{ALLD} = [0, 0, \ldots, 0].$
\textbf{(c)} No full defectors are present within the population.
\textbf{(d)} Population composed only by full cooperators. This state is
the $N$-dimensional vector ${\bf x}^{*}_{ALLC} = [1, 1, \ldots, 1].$
Configurations reported in \textbf{(c)} and \textbf{(d)} represent
cooperative steady states (i.e. $x_v > 0 ~\forall v$).
}}
\label{fig:configurations}
\end{center}
\end{figure}

To fill this gap, in this paper internal mechanisms are introduced in the EGN equation by adding a term $f_{v}$
balancing $\partial \phi_{v}/\partial x_{v}$.
This term is weighted by a parameter $\beta_{v}$ which measures
the inertia of a player with respect to his neighbors' actions \cite{LEONARD2018}.
The extended Self-Regulated EGN equation, hereafter called SR-EGN,
is reported in Figure \ref{fig:eqnBox}. 
Notice that full cooperative ${\bf x}_{ALLC}^{*} = [1, 1, \ldots, 1]^{\top}$,
full defective ${\bf x}_{ALLD}^{*} = [0, 0, \ldots, 0]^{\top}$, and
fuzzy configurations ${\bf x}^{*} = [x_{1}^{*}, x_{2}^{*}, \ldots, x_{N}^{*}]^{\top}$
with $x_{v}^{*} > 0$ for at least one $v$,
are steady states of the SR-EGN equation (see Figure \ref{fig:configurations}).  
Further details on steady states are reported in Appendix \ref{sec:EGNmodel_app}. 

%\noindent {\bf Fig. 2.} 

% TODO
%\\ %TODO: Figure \ref{fig:eqnBox}.

\begin{figure}[!ht]
\begin{center}
\includegraphics[width=0.8\columnwidth]{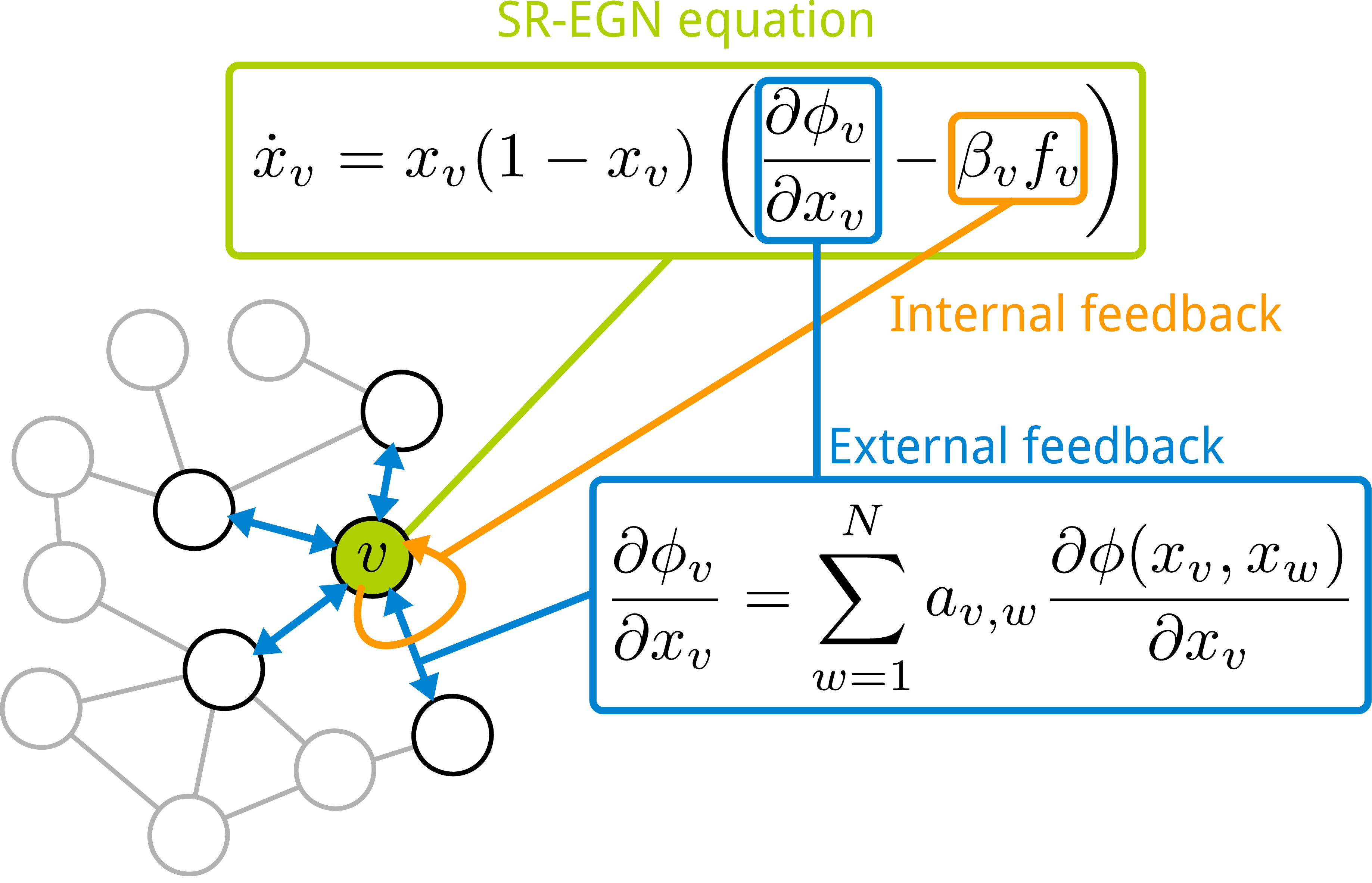}
\caption{\footnotesize{\textbf{SR-EGN equation.} 
The strategy dynamics of player $v$ (green node) is
ruled by the SR-EGN equation (green box). 
It includes two terms: the external feedback term (blue box) 
$\partial \phi_{v}/\partial x_{v}$, accounting for the 
external mechanisms related to game interactions (blue arrows) 
with the $k_{v}$ neighbors (black nodes),
and the internal feedback term (orange box)
$\beta_{v}f_{v}$, modeling the self regulating processes
of node $v$ (the orange self loop). }}\label{fig:eqnBox}
\end{center}
\end{figure}

%

%\noindent {\bf Fig. 1.} 
%

Inspired by self-regulation in animal societies \cite{S1981},
the term $\beta_{v}f_{v}$ is modeled as an internal feedback describing a
virtual game that each individual plays against himself, i.e. a self game. 
This game is characterized by the same parameters $\sigma_{C}$ and $\sigma_{D}$ of the two-player games.
Therefore, the self-regulating function $f(v)$ is written as:
\begin{equation}\label{eqn:fgame}
%f_{v}({\bf x}) = (\sigma_{C} + \sigma_{D})x_{v} - \sigma_{D}.	
f_{v} = (\sigma_{C} + \sigma_{D})x_{v} - \sigma_{D}.	
\end{equation}
Notice that $f_{v}$ is similar to
equation \eqref{eqn:phiDerivative}, where $x_{w}$ has been replaced by $x_{v}$,
thus conceiving the individual $v$ himself as one of his own ``opponents''.
``What kind of 
outcome can I earn if I apply a given strategy to myself?'':
a generic player in our model can be 
aware of the conflicting context where he participates,
and he may know the importance of cooperation
as a primal objective to be pursued. Remarkably, this term models a spontaneous 
learning process, thus representing a time varying feature of each individual.\\
%% --

The complete SR-EGN equation studied in this paper can be rewritten as follows:
\begin{equation}\label{eqn:SREGN2}
\dot{x}_v = \displaystyle x_v(1-x_v)\left[k_v\left((\sigma_C+\sigma_D)\bar{x}_v-\sigma_D\right)-\beta_v\left((\sigma_C+\sigma_D)x_v-\sigma_D\right)\right],
\end{equation}
where $\bar{x}_v = (1/k_v)\sum_{w=1}^N a_{v,w}x_w$ is the average player
resulting from the decisions of all neighboring players of $v$.
Besides $\sigma_C$ and $\sigma_D$, the two fundamental parameters of this model are $k_v$ and $\beta_v$.
$k_{v}$ is the in-degree of player $v$, thus accounting for the influence of the network
(external feedback) on his decision.
The second parameter is the weighting factor $\beta_{v}$ 
modulating self-regulation. When $\beta_{v} = 0$, the individual is somehow ``member of the flock'', 
since his strategy changes only according to the outcome variations of game
interactions with neighbors, embodied by $\partial\phi_{v}/\partial x_{v}$. 
In this case, we recover the standard EGN equation 
\eqref{eqn:EGN-base}, and defection is unavoidable.
Positive values of $\beta_{v}$ represent
an ``aware resistance'' of players to the external feedback.\\

\begin{figure}[!ht] 
\begin{center}
\includegraphics[width=\columnwidth]{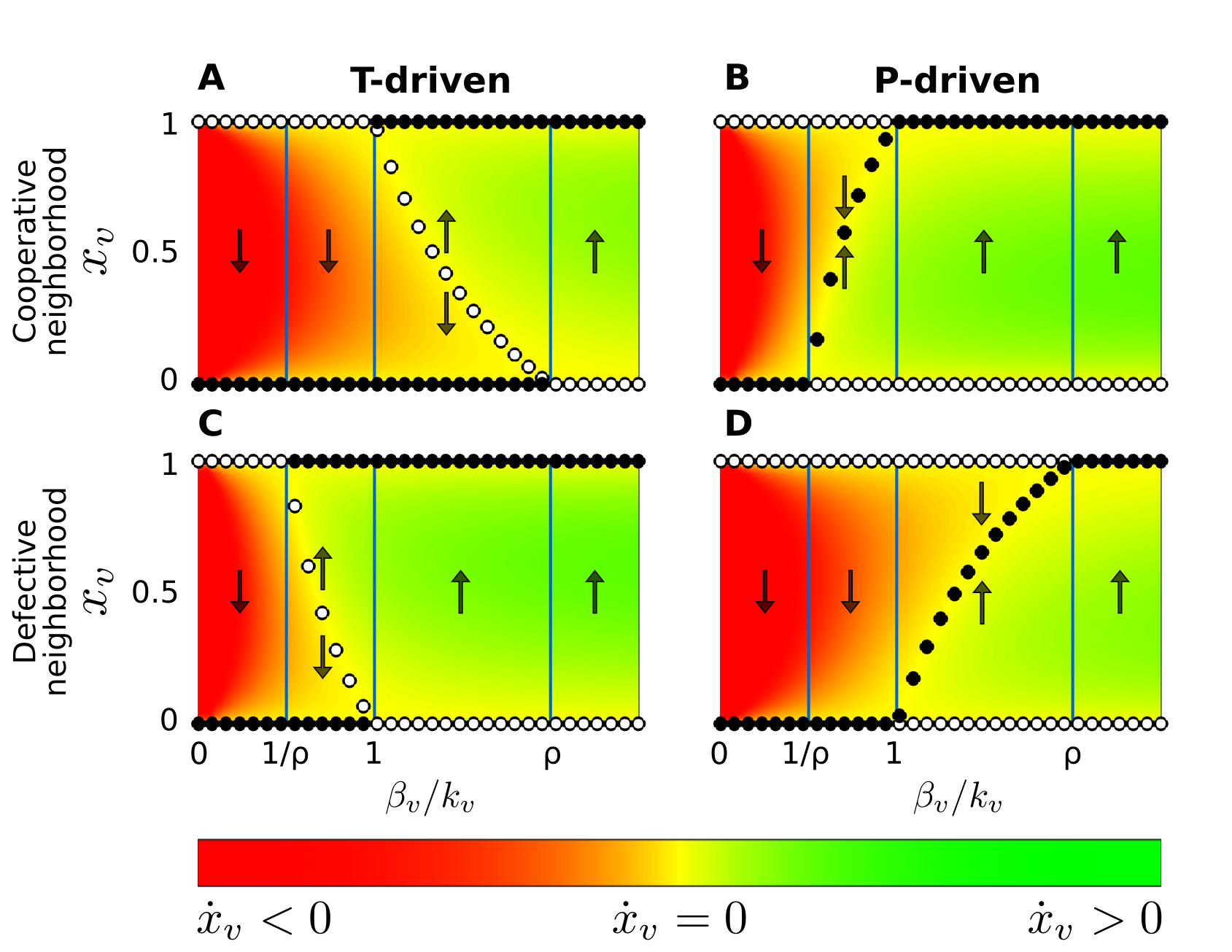}
\caption{\footnotesize{\textbf{Flow of the SR-EGN equation.} 
Colors represent the value of the derivative $\dot{x}_{v}$ as a function of $x_{v}$ and $\beta_{v} / k_{v}$.
Black and white circles correspond to attractive and repulsive steady states for player $v$, respectively.
Arrows are used to schematically highlight the direction of the dynamics.
The flow for a player $v$ connected only to full cooperators ($\bar{x}_{v} = 1$) 
in {\it T-driven} game ($\sigma_{C} = -2$ and $\sigma_{D} = 1$)
and in {\it P-driven} game ($\sigma_{C} = -1$ and $\sigma_{D} = 2$)
are reported in subplots {\bf A} and {\bf B}, respectively.  
Instead, the case of a player connected only to full defectors
($\bar{x}_{v} = 0$) is shown in subplots {\bf C} ({\it T-driven}) and {\bf D} ({\it P-driven}).
Vertical blue lines are depicted for $\beta_{v}/k_{v} = 1/\rho$, $\beta_{v}/k_{v} = 1$ and $\beta_{v}/k_{v} = \rho$.{}
%For $1 <\beta_{v}/k_{v} < \rho$, the $x_{v}=1$ steady state is attractive, while for $\beta_{v}/k_{v} > \rho$ it is also
%globally attractive.%
%
The intermediate steady states corresponding to partial levels of cooperation
are always located between $1/\rho$ and $\rho$.{}
%
 %In {\it T-driven} game their position is shifted on the left
%as long as the number of defectors in the neighborhood increases, thus facilitating cooperative behaviors.
%On the contrary, in {\it P-driven} game these steady states are shifted to the right for ...}
}}\label{fig:flux}
\end{center}
\end{figure}

The main result of this paper is that SR-EGN equation
can explain the emergence of cooperation in a social network.
%Sufficiently high values of $\beta_{v}$ favor the persistence
%of cooperative behavior within the population. 
More specifically,  
when the awareness is stronger than the level of connectivity of each player, i.e. $\beta_{v} > k_{v}$, then
state ${\bf x^{*}}_{ALLC}$ 
is an attractor for the dynamics of the population, as well as a Nash equilibrium \cite{MM2015} of the complete game,
while at the same time total defection ${\bf x^{*}}_{ALLD}$ is repulsive. % (see (Fig. 2.{\bf (b)})) %$Figure \ref{fig:configurations} 
Furthermore, when
\begin{equation}\label{eqn:theorem}
	\beta_{v} > \rho k_{v},
\end{equation} 
where $\rho = \left|\sigma_{C}/\sigma_{D}\right| \geq 1$ for a {\it T-driven} game
and $\rho = \left|\sigma_{D}/\sigma_{C}\right| \geq 1$ for a {\it P-driven} game,
then ${\bf x^{*}}_{ALLC}$ is a global attractor. 
In other words, starting from any initial strategy $x_{v}(t=0) > 0 ~\forall v$,
all players will eventually become full cooperators.
These results are formally proved in the Appendix \ref{sec:survival}.
Notice that convincing individuals with high degree $k_{v}$ to be cooperative requires potentially
large values of $\beta_v$. Anyway, in the following we show that
cooperation may be also achieved for smaller value of $\beta_{v}$.\\

\begin{figure}[!ht]
\begin{center}
\includegraphics[width=\columnwidth]{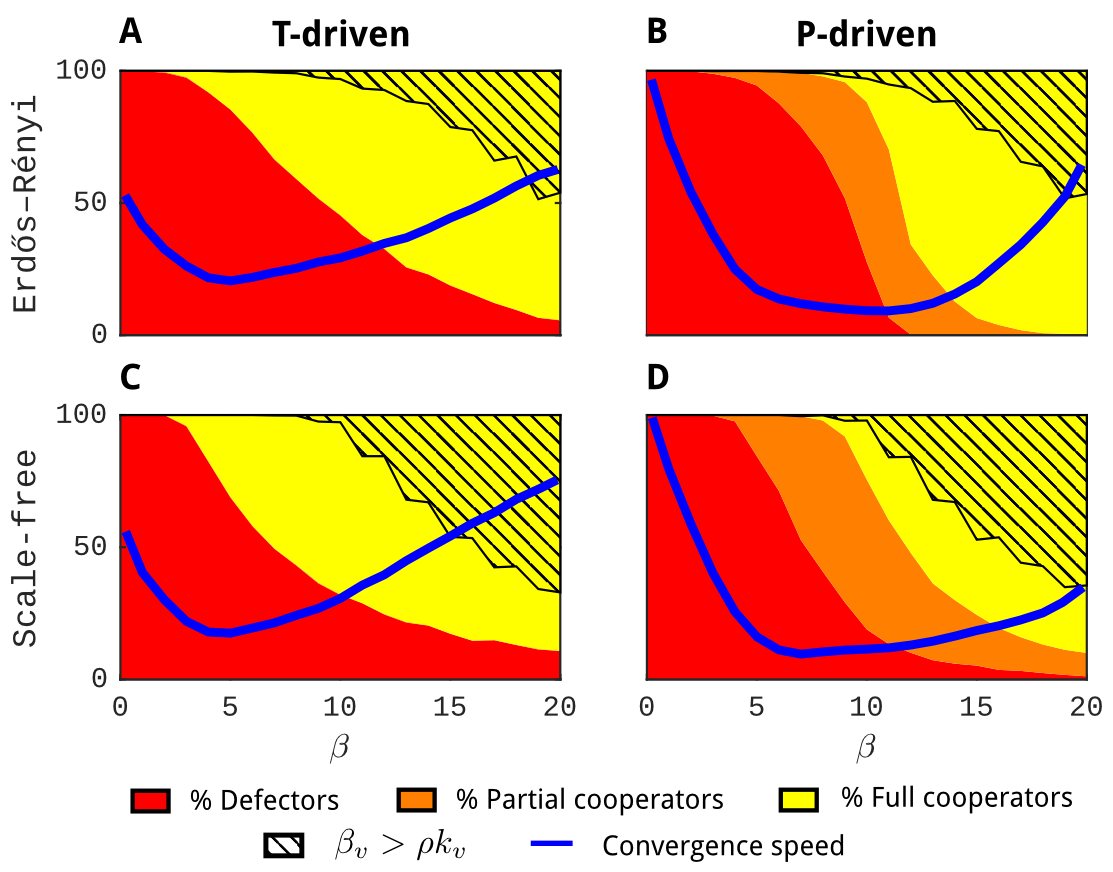}
\caption{\footnotesize{\textbf{Average distribution of strategies and convergence speed.} 
For different graph topologies (ER and SF), 
$500$ graphs with $N = 100$ nodes and average degree $\bar{k} = 10$ are considered.
Initial conditions are randomly chosen in the set $(0, 1)$.
For a {\it T-driven} game ($\sigma_{C} = -2$ and $\sigma_{D} = 1$) and a {\it P-driven} game  
($\sigma_{C} = -1$ and $\sigma_{D} = 2$), and for different values of the parameter $\beta_{v} = \beta \in \{0, \ldots, 20\} ~\forall v$,
the SR-EGN equation is simulated system until a steady state is reached.
The average distribution of strategies of the whole
population, dividing it into 3 subgroups, is reported: defectors in red ($x_{v}^{*} < 10^{-4}$),
partial cooperators in orange ($10^{-4} \leq x_{v}^{*} \leq 1-10^{-4}$) and full cooperators in yellow ($x_{v}^* > 1 - 10^{-4}$). 
The hatched area represent the average percentage of population
for which the theoretical rule $\beta_{v} > \rho k_{v}$ is satisfied. 
The superposed blue lines represent the convergence speed, experimentally estimated as
the inverse of the time required to the population for reaching the steady state.}}\label{fig:avgPerc}
\end{center}
\end{figure}

In Figure \ref{fig:flux} the value of $\dot{x}_{v}$ is reported with colors
as a function of $\beta_{v}/k_{v}$ and $x_{v}$, together with the attracting (black circles)
and repulsive (white circles) steady states for player $v$. Arrows are used to depict the direction
of the dynamics. 
The derivative $\dot{x}_{v}$ is analyzed for both a cooperative ($\bar{x}_{v} = 1$, Figures \ref{fig:flux}{\bf A} and \ref{fig:flux}{\bf B}) 
and a defective ($\bar{x}_{v} = 0$, Figures \ref{fig:flux}{\bf C} and \ref{fig:flux}{\bf D}) neighborhood sets.
The interesting region of all graphs is $\beta_v/k_v < \rho$, in which the values of the parameters are below the 
theoretical threshold \eqref{eqn:theorem}. In all cases, 
the smaller is this ratio, the lower is the probability 
to cooperate, while the higher is the ratio, the higher is the probability to cooperate. 
For intermediate values ($1/\rho < \beta_{v}/k_{v}<\rho$), the formation 
of steady states corresponding to partial levels of cooperation is observed.
These steady states separate the regions where defection or cooperation dominates.
For the {\it T-driven} games (Figures \ref{fig:flux}{\bf A} and \ref{fig:flux}{\bf C}), these new equilibria are 
repulsive, thus creating bistable dynamics which leads player to cooperate or defect
according to their initial conditions.
Moreover, the probability to cooperate raises for increasing values of $\beta_{v}$
or decreasing the in-degree.
Notice that, since the temptation is prominent for this game,
if $\bar{x}_{v}$ moves from $1$ (Figure \ref{fig:flux}{\bf A}) to $0$ (Figure \ref{fig:flux}{\bf C}), 
the region of partial steady states 
exists for lower values of $\beta_{v}/k_{v}$, thus easing cooperation.
On the other hand, 
for {\it P-driven} games (Figures \ref{fig:flux}{\bf B} and \ref{fig:flux}{\bf D}), the fuzzy steady states 
are attractive, thus ensuring at least a certain level of cooperation also in the intermediate region. Interestingly, 
since the punishment is strong, the presence 
of cooperators in the neighborhood facilitates the convergence to a cooperative state. 
Specifically, if $\bar{x}_{v}$ moves from $1$ (Figure \ref{fig:flux}{\bf B}) to $0$ (Figure \ref{fig:flux}{\bf D}),
the region of partial steady states 
exists for higher values of $\beta_{v}/k_{v}$, thus preventing cooperation.
%
%This fact shows that a phenomenon of recruitment on behalf of cooperators is at work, and if 
%enough strong, they can convince some other individuals to cooperate.
This phenomenon will be explained better in the following experiments.
\\

%%%%%%%%%%%%%%%%%%%%%%%%%%%%%%%%%%%
%%%%%%%%%%%%%%%%%%%%%%%%%%%%%%%%%%%
%%%%%%%%%%%%%%%%%%%%%%%%%%%%%%%%%%%

Figure \ref{fig:flux} shows that, from the single player's point of view, cooperation is also 
feasible for values of $\beta_{v}$ below the theoretical threshold \eqref{eqn:theorem}.
%The global emergence of cooperation results from strong requirements,
%since all individuals must satisfy the theoretical results: %TODO (mettere condizioni numeriche?)
%
For completeness, the behavior of the whole population is studied by means of numerical simulations.
We investigate the probability for a population to be cooperative by running $500$ simulations
of different random networks (Erd\"os-R\'enyi (ER) and Scale-Free (SF))
with $N=100$ and average degree $\bar{k}=10$ (and thus the average in-degree is also $10$). 
Moreover, all individuals share the same
self-regulating factor $\beta_v = \beta \in \{0, \ldots, 20\}$. 
This experiment has been repeated for {\it T-driven} and {\it P-driven} games, 
using initial conditions randomly generated for each simulation.

%
%In Figure \ref{fig:avgPerc}
In Figure \ref{fig:avgPerc} the percentages of full defectors (red area), 
partial cooperators (orange area) 
and full cooperators (yellow area)
at steady state are reported as a function of $\beta$.
The fraction of individuals who satisfy the theorem \eqref{eqn:theorem} is highlighted with the hatched pattern.
Notice that the number of defectors decreases by increasing $\beta$.
Consistently with the results shown in Figure \ref{fig:flux}, %Figure \ref{fig:flux},
bistable behavior is observed for {\it T-driven} game (Figures \ref{fig:avgPerc}{\bf A} and \ref{fig:avgPerc}{\bf C}),
for which the population splits into two groups of full defectors and full cooperators.
Partial cooperation (orange area) is present for {\it P-driven} game (Figures \ref{fig:avgPerc}{\bf B} and \ref{fig:avgPerc}{\bf D}).
According to the results shown in Figures \ref{fig:flux}{\bf B} and \ref{fig:flux}{\bf D},
where the presence of cooperative neighborhood fosters cooperation for
single players, we observe the same phenomenon in Figures \ref{fig:avgPerc}{\bf B} and \ref{fig:avgPerc}{\bf D},
extended to the whole population.
In particular, the presence of even few cooperators is able to recruit
their neighboring players to switch their strategies from defection to cooperation. 
Increasing $\beta$, 
these players, together with those satisfying the threshold \eqref{eqn:theorem}, 
are able to recruit to cooperation an increasing number of individuals.
%
%as expected from the proposed theoretical results.
%
%Moreover, a group of full or partial cooperators, which do not satisfy the theorem,
%is present.
%This is particularly significant for the {\it P-driven} game, since players
%with $\beta_{v} > \rho k_{v}$ may help other individuals to become cooperative,
%although they have small values of $\beta_{v}$.
%
The average convergence speed of the system to a steady state,
reported by blue lines, shows a slowdown
of the dynamics for intermediate values of $\beta$.
Specifically, this occurs when a small fraction of cooperators appears in the population,
until a sufficiently large number of individuals start to cooperate,
thus accelerating significantly the dynamics. 
%The slowdown is more persistent for the {\it P-driven} game (Fig. 3{\bf B} and Fig. 3{\bf D}), 
%since a portion of the recruited individuals are partial cooperators (orange area).
%Interestingly, for $\beta \in [10, 15]$, the number of cooperators satisfying the theorem (hatched area) 
%is smaller than the recruited ones (yellow and orange areas).
\\

\begin{figure}[!ht]
\begin{center}
\includegraphics[width=\columnwidth]{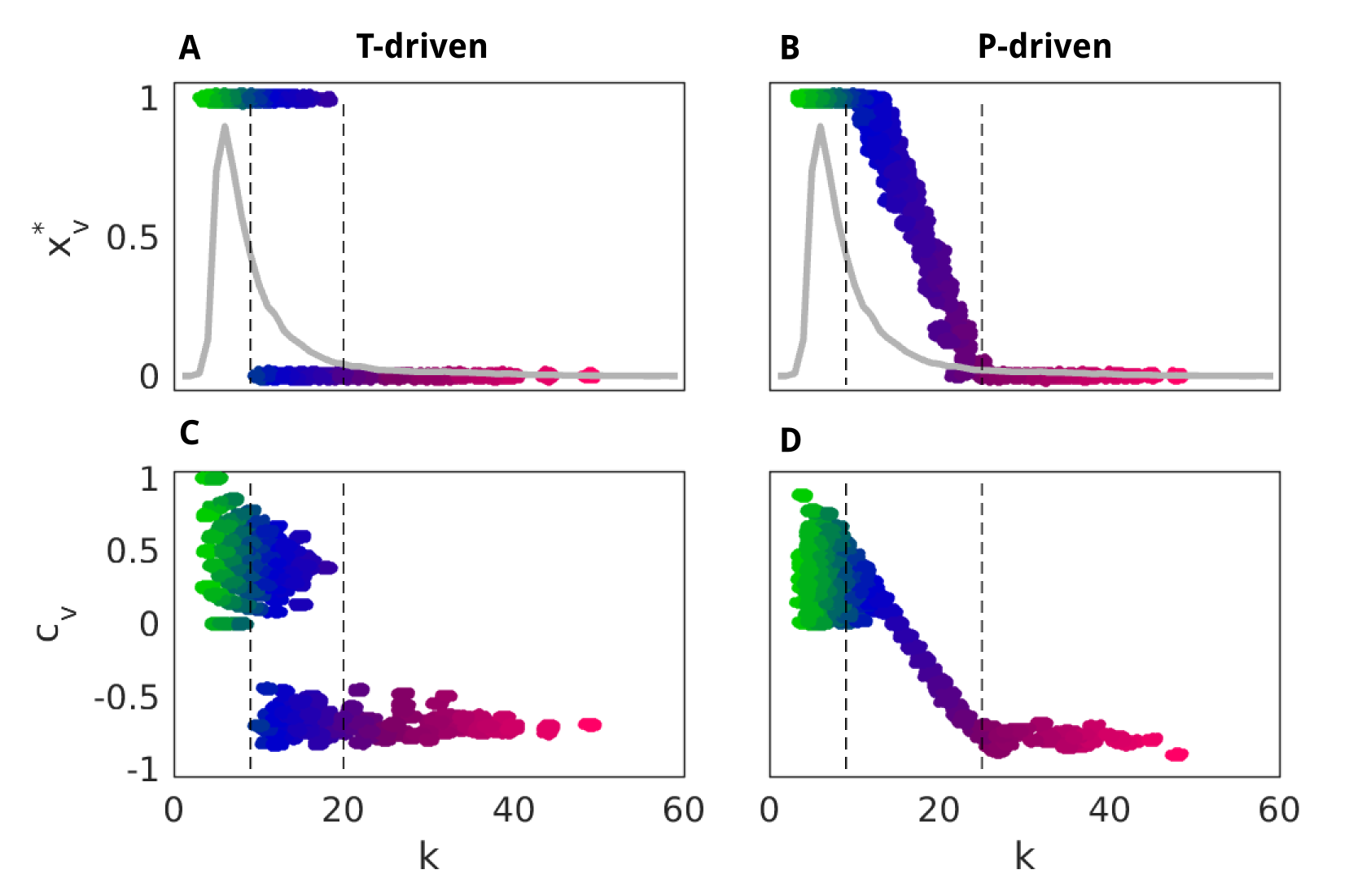}
\caption{\footnotesize{\textbf{Selfishness and altruism within heterogeneous populations.} 
Steady steady state configurations for {\it T-driven} (subplot {\bf A}) and 
{\it P-driven} (subplot {\bf B}) games.
The results have been obtained by simulating the SR-EGN equation
for $500$ realizations of a SF network, composed by $N=100$ individuals,
with average degree $\bar{k} = 10$.
The self regulation parameter is set to $\beta = 15$ for all players. 
The transition from green to blue to magenta indicates 
the social role of players in the graphs, thus allowing to distinguish between 
non central players (green dots), intermediate player (blue dots) and
hubs (magenta dots). 
The vertical dashed lines highlight this distinction. 
The grey lines represent the degree distribution of the networks.
Subplots {\bf C} and {\bf D} report the indicator $c_v$ introduced
in equation \eqref{eq:cv}, for the {\it T-driven} and {\it P-driven} games, respectively.
%
%% NOTA: linee tratteggiate nere: 9-20 per T-driven, 9_25 per P-driven
}}\label{fig:deltaCoop}
\end{center}
\end{figure}

Additionally, the relationship between $\beta_{v}$ and $k_{v}$
is highlighted in Figure \ref{fig:deltaCoop} where $\beta$ is set to $15$, 
and SF networks are used.
For each player, a circle represents the reached steady state as a function of his
in-degree (Figures \ref{fig:deltaCoop}{\bf A} and \ref{fig:deltaCoop}{\bf B}).
In the same subplots, the degree distributions of networks
is depicted in gray.
Players with low degree (green dots), representing non-central individuals,
converge towards a cooperative steady state.
On the other hand, hub players (magenta dots), always prefer defection.
Players with intermediate in-degree (blue dots) show
different behaviors for the {\it T-driven} and {\it P-driven} games.
Specifically, these players split into two subgroups
showing different behaviors (some full cooperators and some full defectors, see Figure \ref{fig:deltaCoop}{\bf A}).
Figure \ref{fig:deltaCoop}{\bf B} shows that these type of players reach a partial level of cooperation.
The distinction of the three groups is highlighted by the dashed vertical black lines. \\

In order to quantify
the difference between the level of cooperation
of player $v$ and the average cooperation of his neighbors at steady state,
the following quantitative indicator is introduced:
\begin{equation}\label{eq:cv}
c_{v} = x_{v}^* - \frac{1}{k_{v}}\sum_{w=1}^{N}a_{v,w}x_{w}^*.
\end{equation}
If $c_{v}>0$, then player $v$ is altruistic, since
his level of cooperation is higher than the average of his neighbors, while $c_{v} < 0$
indicates more selfish behaviors. 
The group of non central players is always more altruistic than 
hubs. 
The intermediate players are again splitted into altruist and selfish
for {\it T-driven} game (Figure \ref{fig:deltaCoop}{\bf C}),
while for {\it P-driven} game (Figure \ref{fig:deltaCoop}{\bf D}), this distinction
vanishes, and a continuous distribution of altruist and selfish persons is observed.\\

%% CONCLUSIONS %%
%This study points out that cooperation may be promoted by introducing
%self-regulating mechanisms 
%in the EGN equation, modeling
%the strategic dynamics of a finite population of individuals arranged
%on a network, and able to choose fully/partially cooperative or defective strategies. %their strategy in the continuous set $[0, 1]$.
%%
%Theoretical results ensure that cooperation emerges in the SR-EGN model whether self-regulation
%is sufficiently stronger than social pressure.
%%
%Besides these theoretical conditions, 
%several numerical experiments show that
Joining the results of Figures \ref{fig:flux}, \ref{fig:avgPerc} and \ref{fig:deltaCoop}, we conclude that
some individuals are more sensitive and aware on
their internal mechanisms, thus becoming cooperative for lower 
self-regulating factors, and exhibiting a more altruistic behavior.
In particular, for the {\it P-driven} game, these receptive individuals
catalyze the others to cooperate. 
%; specifically,
%for  game,
%together with the individuals with high self-regulation,
%they may be able to activate a recruiting mechanism,
%eventually driving the whole population to cooperate.
%
%Moreover, hub individuals are shown to be less cooperative than their neighbors,
%pointing out an attitude toward selfshiness,
%while non-central individuals result to be more inclined to be altruistic
%since they can cooperate more easily compared to their neighbors.  
%

\begin{appendices}%\appendix{The Evolutionary Game equation on Networks (EGN) and self games}
\section{The Evolutionary Game equation on Networks (EGN) and self games}\label{sec:EGNmodel_app}
%\section{Preliminaries}
%%%%%%%%%%%%%%%%%%%%%%%%%%%%%%%%%%%%%%%%%%%%%%%%%%%%%%%%%%%%%%%%%%%%%%%%
%\subsection{The Evolutionary Game equation on Networks (EGN) and self games}\label{sec:EGNmodel}

Let $\playerSet = \{1, 2, \ldots, N\}$ be the set of players. 
Each player is placed in a vertex of a directed graph, defined by the adjacency matrix 
${\bf A} = \{a_{v,w}\} \in \{0, 1\}^{N \times N}$ with ${(v,w) \in \playerSet^{2}}$.
Specifically, $a_{v,w} = 1$ when $v$ is connected with $w$, $0$ otherwise.
It is also assumed that $a_{v,v} = 0$.
The degree of player $v$ is defined as the cardinality of his neighborhood, namely:

$$k_{v} = \sum_{w=1}^{N}a_{v,w}.$$\\

In the literature on evolutionary game theory, it is assumed that, at each time, one individual
uses a pure strategy in a given set while playing games with connected individuals.
For the specific topic of cooperation, these games are often modeled as Prisoner's dilemma games:
the set of pure strategies contains 
only two elements, cooperation ($C$) and defection ($D$), and the outcome is 
described by the following payoff matrix:

$$
{\bf B} = \begin{bmatrix}
	R & S\\
	T & P
\end{bmatrix},
$$

%
%This payoff is the result of all two-players interactions. In particular, 
%in a two-player game,
where $R$ is the reward when both players cooperate, $T$ is the temptation
to defect when the opponent cooperates, $S$ is the sucker's payoff earned by
a cooperative player when the opponent defects, and $P$ is the punishment
for mutual defection. More specifically, for a Prisoner's dilemma game,
the reward is a better outcome than the punishment ($R > P$),
the temptation payoff is higher than the reward ($T > R$),
and the punishment is preferred to the sucker's payoff ($P > S$).
It is useful to define $\sigma_{C} = R-T<0$ and $\sigma_{D} = P-S > 0$ \cite{WEIBULL1995},
allowing us to distinguish two cases: when the effect of $\sigma_{C}$ is stronger than $\sigma_{D}$ ($|\sigma_{C}| > \sigma_{D}$) 
we will refer to a {\it T-driven} game,
while the opposite situation ($|\sigma_{C}| < \sigma_{D}$) is hereafter called {\it P-driven} game.
\\

In the present work, a more realistic scenario is investigated, where one player can 
choose his own level of cooperation, instead of just $C$ or $D$. 
The level of cooperation of a generic player $v$, 
is denoted by the real number $x_{v} \in [0, 1]$; specifically, $x_{v}=1$ stands
for a player exhibiting maximum cooperativeness, while $x_{v} = 0$ represents a full defector. All other
shades (i.e. $x_{v} \in (0,1)$) denote partial levels of cooperation.
In this new framework, when any two connected players $v$ and $w$ take part in a game,
the payoff for $v$ is defined by the continuous bilinear function $\phi : [0, 1] \times [0, 1] \to \mathbb{R}$: %\cite{Weibull1995,SF2007,MM2015,HS2003}
\begin{eqnarray}	
\phi(x_{v}, x_{w}) & = & [x_{v}, 1-x_{v}]{\bf B}\begin{bmatrix}
x_{w} \\
1-x_{w}
\end{bmatrix} \nonumber \\
& = & (R - T + P - S)x_v x_w - (P-S)x_{v} - (P-T)x_{w} + P. \nonumber%\label{eqn:phiTwoPlayers}
\end{eqnarray}

Moreover, notice that $x_v, x_w \in \{0,1\}$, we recover the payoffs $R$, $T$, $S$ and $P$.\\

The total payoff of player $v$ is the sum of all outcomes of two-player games with neighbors.
Formally, the payoff function $\phi_{v} : [0, 1]^{N} \to \mathbb{R}$ is defined as follows:
\begin{equation}\nonumber %\label{eqn:phiNetwork}
	\phi_{v}({\bf x}) = \sum_{w = 1}^{N} a_{v,w}\phi(x_{v}, x_{w}),
\end{equation}
where ${\bf x}$ is the vector of all the $x_{v}$ variables.
Moreover, given the vector ${\bf x}$, we define the following payoff of pure strategies $C$ ($x_{v} = 1$) and $D$ ($x_{v} = 0$):
\begin{equation}\nonumber %\label{eqn:pCpD}
\begin{cases}
	p_{v}^{C}({\bf x}) = \displaystyle\sum_{w = 1}^{N} a_{v,w}\phi(1, x_{w}) = \displaystyle\sum_{w = 1}^{N} a_{v,w}((R-S)x_{w} + S)\\
	p_{v}^{D}({\bf x}) = \displaystyle\sum_{w = 1}^{N} a_{v,w}\phi(0, x_{w}) = \displaystyle\sum_{w = 1}^{N} a_{v,w}((T-P)x_{w} + P)\\
\end{cases}.
\end{equation}
Following \cite{MM2015, IMM2016}, the EGN equation for two-strategy games reads as follows:

\begin{equation}\label{eqn:EGNold}
\dot{x}_{v} = x_{v}(1-x_{v})\Delta p_{v}({\bf x}),
\end{equation}

where

\begin{eqnarray}
	\Delta p_{v}({\bf x}) = p_{v}^{C}({\bf x})-p_{v}^{D}({\bf x}) & = & \displaystyle\sum_{w = 1}^{N} a_{v,w}(R-T+P-S)x_{w} - (P-S)) \nonumber\\
	& = & \displaystyle\sum_{w = 1}^{N} a_{v,w}((\sigma_{C}+\sigma_{D})x_{w} - \sigma_{D}). \nonumber
\end{eqnarray} 

It is clear that the level of cooperation of player $v$ increases (decreases)
when $\Delta p_{v}({\bf x})$ is positive (negative).
In other words, the player $v$ will be more cooperative over time as 
long as the payoff he can earn using the pure strategy $C$ is better
than the payoff he can earn using the pure strategy $D$.\\

This comparative evaluation of the benefits provided by the available strategies
can be represented in an alternative way.
Specifically, suppose that player $v$ is able to appraise whether a change of his strategy $x_{v}$
produces an improvement of his payoff $\phi_{v}$. 
This means that, if the derivative
of $\phi_{v}$ with respect to $x_{v}$ is positive (negative), 
the player would like to increase (decrease) his level of cooperation.
Interestingly, 
%since $\phi_{v}$ is linear with respect to $x_{v}$, 
the following result holds: % \todoa{\`e necessario scrivere la dimostrazione di questo?}:

%\begin{equation}\nonumber %\label{eqn:der}
%\frac{\partial \phi_{v}({\bf x})}{\partial x_{v}} = \Delta p_{v}({\bf x}),
%\end{equation}
\begin{eqnarray}\nonumber %\label{eqn:der}
\frac{\partial \phi_{v}({\bf x})}{\partial x_{v}} & = & \sum_{w=1}^{N}a_{v,w}\frac{\partial \phi(x_{v},x_{w})}{\partial x_{v}} \nonumber \\
& = & \sum_{w=1}^{N}a_{v,w}\left((R-T+P-S)x_{w} - (P-S)\right), \nonumber \\
& = & \sum_{w=1}^{N}a_{v,w}\left((\sigma_{C}+\sigma_{D})x_{w} - \sigma_{D}\right), \nonumber \\
& = & \Delta p_{v}({\bf x}). \nonumber
\end{eqnarray}

Thus, the EGN equation \eqref{eqn:EGNold} can be rewritten as follows:

\begin{equation}\label{eqn:EGN}
\dot{x}_{v} = x_{v}(1-x_{v})\frac{\partial \phi_{v}({\bf x})}{\partial x_{v}}.
\end{equation}

It is worthwhile to notice that, while the replicator equation is used to describe the dynamics of population
where strategies correspond to the phenotypes of the individuals,
its extension on graphs, the EGN equation, is suitable for analyzing 
the dynamics of individuals arranged on a network, which are able to choose their strategies 
in the continuous set $[0, 1]$.\\

The EGN equation \eqref{eqn:EGN}, as well as most of the models presented in the literature, 
assumes that the strategy dynamics
of a generic player $v$ is driven only by external factors.
Indeed, $\frac{\partial \phi_{v}}{\partial x_{v}}$ depends only on the state of neighboring players,
not on the current state $x_v$ of player $v$ himself.
Inspired by mechanisms describing self-regulation in animal societies reported in \cite{S1981},
we overcome this issue by introducing the Self-Regulated EGN equation (SR-EGN);
this new model is obtained by adding a self-regulating term $f_{v}$
to the EGN equation, balancing the external feedback $\frac{\partial \phi_{v}}{\partial x_{v}}$, thus reading as follows:
\begin{equation}\label{eqn:EGN-self}
\dot{x}_{v} = x_{v}(1-x_{v})\left(\frac{\partial \phi_{v}({\bf x})}{\partial x_{v}} - \beta_{v} f_{v}(x_{v})\right),
\end{equation}
where the parameter $\beta_{v}$ is used to tune the effectiveness of the
introduced self-regulation mechanism.
Specifically, we assume that this self-regulation term embodies a game
that a given individual plays against himself.
To describe this self game, consider two generic players which strategies are $y$ and $z$, 
both belonging to the set $[0, 1]$.
As already mentioned, the first player can assess whether a change of his strategy $y$ can lead to an improvement
of the payoff $\phi(y,z)$. In particular, the assessment is based on the sign of the partial derivative
$$\frac{\partial \phi(y,z)}{\partial y} = (\sigma_{C} + \sigma_{D})z-\sigma_{D}.$$
In the particular case of individuals representing both the first and second player at the same time,
the derivative reads as follows: 

\begin{equation}\nonumber
	\left.\frac{\partial \phi(y,z)}{\partial y}\right|_{y=z=x_{v}} = (\sigma_{C} + \sigma_{D})x_{v}-\sigma_{D}.
\end{equation}

Therefore, the self-regulating term is defined as:

\begin{equation}\nonumber
	f_{v}(x_{v}) = \left.\frac{\partial \phi(y,z)}{\partial y}\right|_{y=z=x_{v}} = (\sigma_{C} + \sigma_{D})x_{v}-\sigma_{D}.
\end{equation}

It is worthwhile to notice that the time derivative of $x_{v}$ in \eqref{eqn:EGN-self}
depends on $x_v$ in the term accounting for the self game.
Thus, the self game introduces a \textbf{feedback} mechanism
regulated by the parameter $\beta_{v} \in \mathbb{R}$.
In particular, in equation \eqref{eqn:EGN-self}, $\beta_{v} > 0$ represents a negative feedback, $\beta_{v} < 0$ stands for a positive 
feedback, while $\beta_{v} = 0$ refers to situations where the player $v$ does not play a self game.

%%$${\bf A}' = {\bf A}  - \diag{\boldsymbol{\beta}},$$
%%where 
%%$$\boldsymbol{\beta} = [\beta_{1} ~ \beta_{2} ~ \ldots ~ \beta_{N}].$$

%%The parameters $\beta_{v}$ on the diagonal of {\bf A}'
%%represent the self loops in the network, and thus model the presence of self games.\\

%As pointed out in Section \ref{sec:ECdef}, the study of the emergence of cooperation
%corresponds to the analysis of stability of the steady states ${\bf x}^{*}_{ALLC}$,
%${\bf x}^{*}_{ALLM}$ and  ${\bf x}^{*}_{ALLD}$. 
%In the following, it will be shown that their stability is strongly related to the presence of feedback terms
%induced by the self games. 

%%%%%%%%%%%%%%%%%%%%%%%%%%%%%%%%%%%%%%%%%%%%%%%%%%%%%%%%%%%%%%%%%%%%%%%%
\subsection{Steady states and linearization}\label{sec:ss_lin}

A steady state $\bf{x^*}$ is a solution of equation \eqref{eqn:EGN-self} satisfying  $\dot{x}_{v} = 0 ~\forall v \in \playerSet$. 
In order to be \textbf{feasible}, the components of a steady state must belong to the set $[0, 1]$.
Formally, the set of feasible steady states is:
$$\Theta = \{{\bf x}^{*} \in \mathbb{R}^{N} : \dot{x}_{v}^{*} = 0 \wedge x_{v}^{*} \in [0,1] ~ \forall v \in \playerSet \}.$$
It is clear that all points such that for all $v$, $x_{v}^{*} = 0$ or $x_{v}^{*} = 1$
are in the set $\Theta$. They are the $2^{N}$ pure steady states. 
We remark that set $\Theta$ may contain also other steady states, exhibiting fuzzy levels of cooperation. 
Particularly relevant are the pure steady states 

$${\bf x}^{*}_{ALLC} = [1, 1, \ldots, 1]^{\top},$$

and

$${\bf x}^{*}_{ALLD} = [0, 0, \ldots, 0]^{\top}.$$

Indeed, they represent a population composed by full cooperators and full defectors, respectively,
and thus they describe the spread of cooperation and its extinction in a given population. 
\\

The dynamical properties of these two pure steady states is fundamental for the emergence of cooperation. 
In particular, their stability can be analyzed by linearizing system \eqref{eqn:EGN-self}
near them.

The Jacobian matrix of system \eqref{eqn:EGN}, ${\bf J}({\bf x}) = \{j_{v,w}({\bf x})\}$, is defined as follows:
\begin{equation}\nonumber %\label{eqn:jacob}
j_{v,w}({\bf x}) = \frac{\partial \dot{x}_{v}}{\partial x_{w}} = \begin{cases}
x_v(1-x_v)(\sigma_{C} + \sigma_{D}), & \text{if}~a_{v,w} = 1 \\
~\\
(1-2x_{v})\displaystyle\left(\frac{\partial \phi_{v}({\bf x})}{\partial x_{v}} - \beta_{v} f_{v}(x_{v})\right)-\beta_{v}x_{v}(1-x_{v})(\sigma_{C} + \sigma_{D})  , & \text{if}~w = v \\
~\\
0, & \text{otherwise}
\end{cases}.
\end{equation}

It is easy to show that the Jacobian matrix reduces to a diagonal one for both ${\bf x}^{*}_{ALLC}$ and ${\bf x}^{*}_{ALLD}$.
%Indeed, off-diagonal entries are $0$ whenever $v$ is not connected to $w$, are multiplied by $x_{v}(1-x_{v})$, which is always $0$ for both ${\bf x}^{*}_{ALLC}$ and ${\bf x}^{*}_{ALLD}$ 
Moreover, observe that:

\begin{eqnarray}
	\frac{\partial \phi_{v}({\bf x}^{*}_{ALLC})}{\partial x_{v}} & = & \sum_{w=1}^{N}a_{v,w}((\sigma_{C}+\sigma_{D}) \cdot 1 - \sigma_{D}) = \sum_{w=1}^{N}a_{v,w}\sigma_{C} = k_{v} \sigma_{C}, \nonumber\\
	\frac{\partial \phi_{v}({\bf x}^{*}_{ALLD})}{\partial x_{v}} & = & \sum_{w=1}^{N}a_{v,w}((\sigma_{C}+\sigma_{D}) \cdot 0 - \sigma_{D}) = -\sum_{w=1}^{N}a_{v,w}\sigma_{D} = -k_{v} \sigma_{D}, \nonumber\\
	f_{v}(1) & = & (\sigma_{C}+\sigma_{D} \cdot 1 - \sigma_{D}) = \sigma_{C}, \nonumber\\
	f_{v}(0) & = & (\sigma_{C}+\sigma_{D} \cdot 0 - \sigma_{D}) = -\sigma_{D}. \nonumber
\end{eqnarray}
Therefore, we have that:
$$j_{v,v}({\bf x}^{*}_{ALLC}) = (1-2 \cdot 1)\displaystyle\left(\frac{\partial \phi_{v}({\bf x}^{*}_{ALLC})}{\partial x_{v}} - \beta_{v} f_{v}(1)\right) = -\sigma_{C}(k_{v} - \beta_{v})$${}
for ${\bf x}^{*}_{ALLC}$, and
$$j_{v,v}({\bf x}^{*}_{ALLD}) = (1-2 \cdot 0)\displaystyle\left(\frac{\partial \phi_{v}({\bf x}^{*}_{ALLD})}{\partial x_{v}} - \beta_{v} f_{v}(0)\right) = -\sigma_{D}(k_{v} - \beta_{v})$${}
for ${\bf x}^{*}_{ALLD}$.

%%%%%%%%%%%%%%%%%%%%%%%%%%%%%%%%%%%%%%%%%%%%%%%%%%%%%%%%%%%%%%%%%%%%%%%%
\section{Emergence of cooperation in the EGN equation with self-regulations}\label{sec:survival}

%Following \cite{SF2007, Wang2015,Nowak2004, SANTOS2006, Nowak2006} \todoa{rivedere queste citazioni}, 
The emergence of cooperation
is reached when all the members of a social network turn their strategies
to cooperation.
Therefore, the asymptotic stability of ${\bf x}^{*}_{ALLC}$, as well
as the instability of ${\bf x}^{*}_{ALLD}$, has a fundamental role in this context.
In order to study the stability of steady states ${\bf x}^{*}_{ALLC}$
and ${\bf x}^{*}_{ALLD}$, we start by analyzing their
linear stability. Moreover, an appropriate Lyapunov function is proposed, 
which prove that ${\bf x}^{*}_{ALLC}$ is also globally asymptotically stable,
thus guarantying the emergence of cooperation.

%%%%%%%%%%%%%%%%%%%%%%%%%%%%%%%%%%%%%%%%%%%%%%%%%%%%%%%%%%%
\subsection{Asymptotic stability of ${\bf x}^{*}_{ALLC}$}%\label{app:theorems} % and ${\bf x}^{*}_{ALLD}$}
Recall that the spectrum of ${\bf J}({\bf x}^{*})$ characterizes the linear 
stability of any steady state ${\bf x}^{*}$ \cite{Strogatz2001}. Therefore,
the role of the eigenvalues of the Jacobian matrix ${\bf J}({\bf x}^{*})$ 
is fundamental to tackle the problem of the emergence of cooperation.\\

The following results hold.

%% THEOREM 1 %%%%%%%%%%%%%%%%%%%%%%%%%%%%%%%%%%%%%%%%%%%%%%%%%%%%%%%%%%%
\begin{theorem}\label{th1}
If $\beta_{v} > k_{v} ~\forall v \in \playerSet$,
then ${\bf x}^{*}_{ALLC}$ is asymptotically stable.
\end{theorem}

\begin{proof}
As shown before, the Jacobian matrix evaluated for ${\bf x}^{*}_{ALLC}$ is diagonal. Then,
the elements on the diagonal of the Jacobian matrix correspond to its eigenvalues and they are defined as follows:
\begin{equation}\label{eqn:allCeigenvalues}
j_{v,v}({\bf x}^{*}_{ALLC}) = \lambda_{v} = -\sigma_{C}(k_{v} - \beta_{v}). \nonumber
\end{equation}

Using the fact that $\beta_{v} > k_{v} ~\forall v \in \playerSet$, and $\sigma_C < 0$, all the eigenvalues are negative. Thus,
${\bf x}^{*}_{ALLC}$ is asymptotically stable.
\end{proof}
%

%We remark that Theorem 2 stated in \cite{MM2015} assures that any asymptotically stable pure steady state
%is also a Nash equilibrium, and viceversa. Then, under the hypotheses of Theorem \ref{th1}, ${\bf x}^{*}_{ALLC}$
%is a Nash equilibrium.
%
%% THEOREM 2 %%%%%%%%%%%%%%%%%%%%%%%%%%%%%%%%%%%%%%%%%%%%%%%%%%%%%%%%%%%
\begin{theorem}\label{th2}
If $\exists v \in \playerSet : \beta_{v} > k_{v}$, 
then ${\bf x}^{*}_{ALLD}$ is unstable.
\end{theorem}

\begin{proof}
The eigenvalues of the Jacobian matrix relative to the steady state ${\bf x}^{*}_{ALLD}$ are
\begin{equation}\nonumber %\label{eqn:allDeigenvalues}
j_{v,v}({\bf x}^{*}_{ALLD}) = \lambda_{v} = -\sigma_{D}(k_{v} - \beta_{v}).
\end{equation}
If the hypothesis of the theorem are fulfilled, since $\sigma_D>0$, then there is at least one
positive eigenvalue, implying that ${\bf x}^{*}_{ALLD}$ is an unstable steady state.
\end{proof}
%
%The results of the previous Theorem, jointly with the Theorem 2 enunciated in \cite{MM2015}, imply that 
%${\bf x}^{*}_{ALLD}$ is not a Nash equilibrium of the underlying game.\\

These results are summarized as follows: defection dominates over cooperation. 
Then, if the system does not present any internal feedback mechanism (i.e. $\beta_{v} = 0 ~\forall v \in \playerSet$), 
the whole social network will converge to ${\bf x^{*}}_{ALLD}$ (cooperation vanishes).
Anyway, using $\beta_{v} > k_{v}$ for all the members of the population, ${\bf x^{*}}_{ALLD}$
is destabilized and ${\bf x^{*}}_{ALLC}$ becomes attractive.
%%%%%%%%%%%%%%%%%%%%%%%%%%%%%%%%%%%%%%%%%%%%%%%%%%%%%%%%%%%%%%%%%%%%%%%%
\subsection{Global asymptotic stability of ${\bf x}^{*}_{ALLC}$}

Theorems \ref{th1} and \ref{th2} prove that under suitable condition, ${\bf x}^{*}_{ALLC}$ is asymptotically stable
and ${\bf x}^{*}_{ALLD}$ is unstable. Anyway, this is not sufficient to prove
emergence of cooperation. Indeed, there can be some other steady states in $\Theta$
which may be also attractive. 
Nevertheless, a Lyapunov function for the steady state ${\bf x}^{*}_{ALLC}$
on the set ${\bf x} \in (0,1]^{N}$ can be found \cite{KA2002}.\\

Adapting the approach presented in \cite{WEIBULL1995,HS1998} to the SR-EGN equation, we consider the following function:
\begin{equation}\nonumber %\label{eq:lyapfun2}
V({\bf x}) = -\sum_{v=1}^{N}\log(x_{v}),
\end{equation}
for ${\bf x} \in (0,1]^{N}$. Notice that $V({\bf x}^{*}_{ALLC}) = 0$,
and $V({\bf x}) > 0 ~\forall {\bf x} \neq {\bf x}^{*}_{ALLC}$.\\

Moreover, the time derivative of $V({\bf x})$ is defined as follows:
\begin{eqnarray}
\dot{V}({\bf x}) = \frac{\partial V({\bf x})}{\partial t} & = & \sum_{v=1}^{N}\frac{\partial V({\bf x})}{\partial x_{v}}\dot{x}_{v} = \nonumber\\	
& = & -\sum_{v=1}^{N}\frac{1}{x_{v}}x_{v}(1-x_{v})\left(\frac{\partial \phi_{v}({\bf x})}{\partial x_{v}} - \beta_{v} f_{v}(x_{v})\right) = \nonumber\\	
& = & \sum_{v=1}^{N}(x_{v} - 1)\left(\frac{\partial \phi_{v}({\bf x})}{\partial x_{v}} - \beta_{v} f_{v}(x_{v})\right). \label{eq:lyapALLCtder}	
 %= \nonumber\\	
%& = & \sum_{v=1}^{N}(x_{v} - 1)\left(\sum_{w=1}^{N}a_{v,w}\left((\sigma_{C}+\sigma_{D})x_{w} - \sigma_{D}\right) - \beta_{v}\left(\left(\sigma_{C}+\sigma_{D}\right)x_{v} - \sigma_{D}\right)\right) = \nonumber\\	
%& = & (\sigma_{C}+\sigma_{D})\sum_{v=1}^{N}(x_{v} - 1)\left(\sum_{w=1}^{N}a_{v,w}\left(x_{w} - m\right) - \beta_{v}\left(x_{v} - m\right)\right) = \nonumber\\	
%& = & (\sigma_{C}+\sigma_{D})\sum_{v=1}^{N}(x_{v} - 1)q_{v}({\bf x}),\label{eq:lyapALLCcond}	
\end{eqnarray}
Finally, it is clear that $\dot{V}({\bf x}_{ALLC}^{*}) = 0$. \\

Starting from these premises, if $\dot{V}({\bf x}) < 0$ for all ${\bf x} \in (0, 1]^{N} \setminus \{{\bf x}^{*}_{ALLC}\}$, then $V({\bf x})$ is a Lyapunov function.\\

Let's introduce the following quantities:

\begin{equation}\label{eqn:minPhi}
	\psi = \inf_{y \in (0, 1]}(\sigma_{C}+\sigma_{D})y - \sigma_{D},
\end{equation}

\begin{equation}\label{eqn:maxPhi}
	\xi = \sup_{y \in (0, 1]}(\sigma_{C}+\sigma_{D})y - \sigma_{D},
\end{equation}
and
$$\rho = \frac{\psi}{\xi}.$$
Notice that, since $\sigma_{C} < 0$ and $\sigma_{D} > 0$, then both $\psi$ and $\xi$ are negative. 
Indeed,
\begin{itemize}
\item for a {\it T-driven} game, since $|\sigma_{C}| > \sigma_{D}$, then $\psi = \sigma_{C}$ and $\xi = -\sigma_{D}$;
\item for a {\it P-driven} game, since $|\sigma_{C}| < \sigma_{D}$, then $\psi = -\sigma_{D}$ and $\xi = \sigma_{C}$.
\end{itemize}

Then, we conclude that $\rho$ is a positive number if both cases, and, in particular, 
it is:

$$\rho = \frac{\psi}{\xi} = \begin{cases}
\displaystyle\left|\frac{\sigma_{C}}{\sigma_{D}}\right| & \text{for {\it T-driven} games} \\
\\
\displaystyle\left|\frac{\sigma_{D}}{\sigma_{C}}\right| & \text{for {\it P-driven} games}.
\end{cases}$$

%as a consequence $\rho$ is a positive number.\\

The following result holds.

%% THEOREM %%
\begin{theorem}\label{th3}
If $\beta_{v} > \rho k_{v} ~\forall v \in \playerSet$, then
$V({\bf x})$ is a Lyapunov function.
\end{theorem}

\begin{proof}
It is straightforward to see that:

\begin{equation}\label{eqn:cond1}
	\frac{\partial \phi_{v}({\bf x})}{\partial x_{v}} = \sum_{w=1}^{N}a_{v,w}((\sigma_{C}+\sigma_{D})x_{w} - \sigma_{D}) \geq \sum_{w=1}^{N}a_{v,w} \psi = k_{v} \psi.
\end{equation}
Similarly, since $\beta_{v} > 0$, we get that:

\begin{equation}\label{eqn:cond2}
	\beta_{v} f_{v}(x_{v}) = \beta_{v}((\sigma_{C}+\sigma_{D})x_{v} - \sigma_{D}) \leq \beta_{v} \xi.
\end{equation}

Joining \eqref{eqn:cond1} and \eqref{eqn:cond2} together, we get that:

\begin{equation}\nonumber %\label{eqn:cond3}
	\frac{\partial \phi_{v}({\bf x})}{\partial x_{v}} - \beta_{v} f_{v}(x_{v}) \geq k_{v} \psi- \beta_{v} \xi.
\end{equation}

Moreover, notice that:

$$\beta_{v} > \rho k_{v} \Rightarrow \beta_{v} > \frac{\psi}{\xi}k_{v} \Rightarrow k_{v} \psi- \beta_{v} \xi > 0,$$

and hence

\begin{equation}\label{eqn:cond4}
	\frac{\partial \phi_{v}({\bf x})}{\partial x_{v}} - \beta_{v} f_{v}(x_{v}) > 0 ~\forall v \in \playerSet.
\end{equation}

According to equations \eqref{eq:lyapALLCtder} and \eqref{eqn:cond4}, since $x_{v}-1 < 0$ for all ${\bf x} \in (0, 1]^{N} \setminus \{{\bf x}^{*}_{ALLC}\}$,
we guarantee that $\dot{V}({\bf x}) < 0$ for all ${\bf x} \in (0, 1]^{N} \setminus \{{\bf x}^{*}_{ALLC}\}$. Hence, $V({\bf x})$ is a Lyapunov function.
\end{proof}
%% END THEOREM %%

\begin{figure}[!ht] 
\begin{center}
\includegraphics[width=\columnwidth]{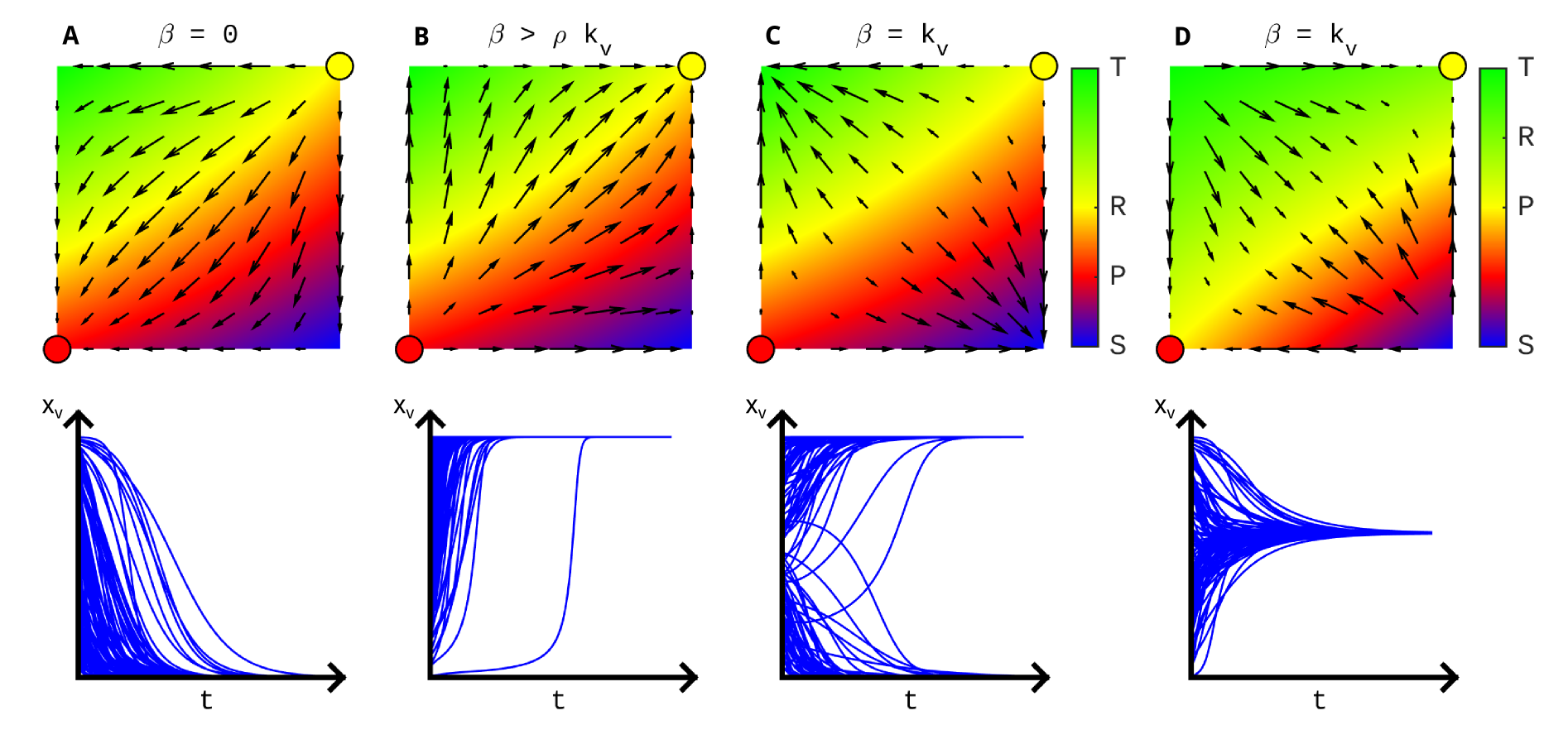}
\caption{\footnotesize{\textbf{Payoff, flow and dynamics.} 
Value of the two-player game payoff $\phi(x_v, x_w)$,  %\eqref{eqn:phiTwoPlayers}
and corresponding flow of the SR-EGN equation in the phase space $(x_{v}, x_{w})$
%in a bidimensional phase space 
for different values of
$\beta_{v}$. 
Subplots {\bf A}, {\bf B} and {\bf C}:
{\it T-driven} game ($\sigma_{C} = -2$ and $\sigma_{D} = 1$);
subplot {\bf D}: {\it P-driven} game ($\sigma_{C} = -1$ and $\sigma_{D} = 2$).
Red and yellow circles represent the steady states ${\bf x}_{ALLD}^{*}$
and ${\bf x}_{ALLC}^{*}$, respectively.
For each case, the corresponding dynamics for a population of $N=100$ players 
arranged on a scale-free network with average degree $\bar{k} = 10$ is depicted below, 
where initial conditions are randomly chosen in the set $(0,1)$.
Results for $\beta_{v} = 0$ and $\beta_{v} > \rho k_{v}$ in the {\it T-driven}
game are similar to the {\it P-driven} game, thus
not reported here.}}\label{fig:flux2}
\end{center}
\end{figure}
%%%%%%%%%%%%%%%%%%%%
%\section{Altro}
%Figure \ref{fig:flux} 
Figure \ref{fig:flux2} compares the above theoretical results 
to the standard case ($\beta_{v} = 0$),
by showing the flow
in phase space of EGN and SR-EGN equations,
and the corresponding time course of
the cooperation level
for a population of $N=100$ players organized on a scale-free network.
Figures \ref{fig:flux2}{\bf A} and \ref{fig:flux2}{\bf B} show that all players are attracted by full defection (EGN) and full cooperation (SR-EGN), respectively.
The same scheme has been used to investigate the marginal case $\beta_v = k_v$ in the {\it T-driven} (Figure \ref{fig:flux2}{\bf C}) 
and {\it P-driven} (Figure \ref{fig:flux2}{\bf D}) games.
In the first case, some players are attracted by full cooperation and 
some others by full defection, while in the second we observe the presence
%of a line of attracting homogeneous partially cooperative steady states.
of an attracting line of partially cooperative steady states, where all players share
the same level of cooperation.
In this case, the different level reached depends on the initial conditions.\\

%The global emergence of cooperation results from strong requirements,
%since all individuals must satisfy the theoretical results: %TODO (mettere condizioni numeriche?)
%convincing individuals with high degree $k_{v}$ to be cooperative requires potentially
%large values of $\beta_v$.
%

\end{appendices}

\section*{Acknowledgments}
CM was partially supported by grant 313773/2013-0 of the Science without Borders Program of CNPq/Brazil.

\end{document}